\documentclass[letterpaper, 11pt]{article}
\usepackage[margin=1in,paperheight=11in,paperwidth=8.5in]{geometry}
\usepackage{amsthm}
\usepackage{amsmath,amssymb}
\usepackage[utf8]{inputenc} 
\usepackage[T1]{fontenc}    
\usepackage{hyperref}       
\usepackage{url}            
\usepackage{booktabs}       
\usepackage{amsfonts}       
\usepackage{nicefrac}       
\usepackage{microtype}      
\usepackage{palatino}
\usepackage{mathpazo}
\usepackage{graphicx}
\usepackage{xcolor}
\usepackage[linesnumbered,ruled,vlined]{algorithm2e}
\SetKwInput{Input}{Input}
\SetKwInput{Output}{Output}
\usepackage{lipsum}

\newtheorem{thm}{Theorem}

\newtheorem{lemma}{Lemma}[section]

\newtheorem{theorem}[lemma]{Theorem}
\newtheorem{cor}[lemma]{Corollary}

\newtheorem{rem}[lemma]{Remark}

\newcommand{\re}{\begin{rem}\rm}
\newcommand{\mar}{\end{rem}}

\def\<{\langle}
\def\>{\rangle}
\def\be{\begin{equation*}}
\def\ee{\end{equation*}}
\def\bea{\begin{eqnarray*}}
\def\eea{\end{eqnarray*}}

\newcommand{\comment}[1]{}

\newcommand{\pl}{\hspace{.1cm}}

\newcommand{\C}{{\mathcal C}}

\newcommand{\cB}{{\mathcal B}}

\newcommand{\cH}{{\mathcal H}}

\newcommand{\eps}{\varepsilon}

\newcommand{\ten}{\otimes}

\newcommand{\ket}[1]{\ensuremath{\left|#1\right\rangle}}
\newcommand{\op}[2]{|#1\rangle \langle #2|}

\newcommand{\tr}{\mathop{\mathrm{tr}}\nolimits}

\newcommand{\norm}[2]{\parallel \! #1 \! \parallel_{#2}}

\begin{document}

\title{Sample optimal tomography of quantum Markov chains}

\author{Li Gao \thanks{Department of Mathematics, University of Houston, Houston, Texas, USA.}\ \ and\ \
Nengkun Yu \thanks{
Centre for Quantum Software and Information, Faculty of Engineering and Information Technology, University of Technology, Sydney, NSW 2007, Australia.}}


\maketitle


\begin{abstract}
A state on a tripartite quantum system $\cH_{A}\ten \cH_{B}\ten\cH_{C} $ forms a Markov chain, i.e., quantum conditional independence, if it can be reconstructed from its marginal on $\cH_{A}\ten \cH_{B}$ by a quantum operation from $ \cH_{B}$ to $ \cH_{B}\ten\cH_{C} $ via the famous Petz map: a quantum Markov chain $\rho_{ABC}$ satisfies $\rho_{ABC}=\rho_{BC}^{1/2}(\rho_B^{-1/2}\rho_{AB}\rho_B^{-1/2}\otimes id_C)\rho_{BC}^{1/2}$.

In this paper, we study the robustness of the Petz map for different metrics, i.e., the closeness of marginals implies the closeness of the Petz map outcomes. The robustness results are dimension-independent for infidelity $\delta$ and trace distance $\epsilon$. The applications of robustness results are
\begin{itemize}
\item The sample complexity of quantum Markov chain tomography, i.e., how many copies of an unknown quantum Markov chain are necessary and sufficient to determine the state, is $\tilde{\Theta}(\frac{(d_A^2+d_C^2)d_B^2}{\delta})$, and $\tilde{\Theta}(\frac{(d_A^2+d_C^2)d_B^2}{\epsilon^2}) $.
\item The sample complexity of quantum Markov Chain certification, i.e., to certify whether a tripartite state equals a fixed given quantum Markov Chain $\sigma_{ABC}$ or at least $\delta$-far from $\sigma_{ABC}$, is ${\Theta}(\frac{(d_A+d_C)d_B}{\delta})$, and ${\Theta}(\frac{(d_A+d_C)d_B}{\epsilon^2})$.
\item $\tilde{O}(\frac{\min\{d_Ad_B^3d_C^3,d_A^3d_B^3d_C\}}{\epsilon^2})$ copies to test whether $\rho_{ABC}$ is a quantum Markov Chain or $\epsilon$-far from its Petz recovery state. The bound is better than the standard tomography of general $\rho_{ABC}$ with $d_A\gg d_Bd_C$. In other words, tomography is not always necessary for testing quantum conditional independence.
\end{itemize}

We generalized the tomography results into multipartite quantum system $\otimes_{i=1}^n \cH_i$ by showing $\tilde{O}(\frac{n^2\max_{i} \{d_i^2d_{i+1}^2\}}{\delta})$ copies for infidelity $\delta$ are enough for $n$-partite quantum Markov chain tomography with $d_i$ being the dimension of the $i$-th subsystem.

We also prove the continuity of the Petz map for general quantum channels in $\ell_2$ distance, which may be of independent interest.

\end{abstract}


\section{Introduction}
Estimating how many copies of an unknown quantum state are needed to identify the quantum state is a fundamental question in physics experiments. Quantum property testing provides a primary tool for achieving reliable control of quantum devices. Various quantum property testing problems occur during the preparation of an experimental setup \cite{MdW13}, such as quantum state certification--whether a state is close to a target state, purity--whether a state is pure, and entanglement verification--whether a state is entangled.

Quantum property testing is the natural extension of the classical distribution property testing problem, which is a fundamental and active research topic in statistics. The study of quantum property testing has attracted much attention recently \cite{MdW13,NIPS2018_8111,10.1145/3313276.3316378,doi:10.1137/18M120275X,10.1145/3406325.3451109,https://doi.org/10.48550/arxiv.2009.04610,10.1007/978-3-030-61792-9_30,https://doi.org/10.48550/arxiv.2206.05265}.
Among many scenarios of quantum property testing, quantum state tomography is a remarkable and central problem.
Quantum state tomography aims to reconstruct the quantum states of physical systems. For an unknown $d$-dimensional mixed quantum state $\rho$ with a rank less than $r$, how many copies of $\rho$ are necessary and sufficient to generate an output of a good approximation of ${\rho}$ with high probability? Provided general joint measurements are allowed, i.e., measurements on $\rho^{\otimes n}$ for any $n$, the sample complexity of state tomography is $\Theta(\frac{d^2}{\epsilon^2})$ for trace distance $\epsilon$, or $\tilde{\Theta}(\frac{dr}{\delta})$ for infidelity $\delta$ (see \cite{HHJ+16,OW16,OW17}). 

Quantum state certification, or general quantum identity testing, has also been actively studied \cite{PhysRevLett.106.230501,PhysRevLett.107.210404,AGKE15,9318006,yu2019quantum,fanizza2021testing,https://doi.org/10.48550/arxiv.2009.11518,OW15}. This problem asks whether an unknown state is identical to the other state, known or unknown. \cite{BOW17} showed the sample complexity of the quantum identity testing (and state certification problem) is $\Theta(\frac{d}{\epsilon^2})$.

For many-body quantum systems, the state space dimension grows exponentially concerning the number of parties. Such growth is a fundamental challenge in performing quantum property testing. On the other hand, the quantum states of physical interest usually have special properties or are within a specific structure. Among all the quantum states of physical interest, quantum Markov chains are an important class. A quantum state $\rho_{ABC}$ is a quantum Markov chain if it can be reconstructed from its marginal $\rho_{AB}$ by a quantum operation from $\cH_{B}$ to $ \cH_{B}\ten\cH_{C}$. Due to the importance of quantum Markov Chains, several equivalence conditions have been characterized.
\begin{itemize}
\item $\rho_{ABC}$ is a quantum Markov chain, i.e., there exists a quantum operation $\mathcal{N}_{B\mapsto BC}$ from the $\cH_B$ to
the $\cH_{B\otimes C}$ system such that
$\rho_{ABC} = \mathcal{N}_{B\mapsto BC}(\rho_{AB})$.
\item $\rho_{ABC}$ can be reconstructed via the Petz map \cite{Petz_1996}, i.e., \[\rho_{ABC}=\rho_{BC}^{1/2}(\rho_B^{-1/2}\rho_{AB}\rho_B^{-1/2}\otimes id_C)\rho_{BC}^{1/2}\ .\]
\item The quantum conditional mutual information of $\rho_{ABC}$ is 0, i.e. \cite{Petz_1996}, \[I(A:C|B):=S(AB)_{\rho}+S(BC)_{\rho}-S(B)_{\rho}-S(ABC)_{\rho}=0\ .\] Here, $S(\rho) = -\tr(\rho \log_2 \rho)$ is the von Neumann entropy.
\item $\rho_{ABC}=\oplus_k p_k \rho_{AB_{L,k}}\otimes \rho_{CB_{R,k}}$
for a probability distribution $\{p_{k}\}$ and states $\rho_{AB_{L,k}}\in \cH_{A}\otimes\cH_{B_{L,k}}$ and $\rho_{CB_{R,k}}\in \cH_{C}\otimes\cH_{B_{L,k}}$ where $\cH_B =\oplus_k (\cH_{B_{L,k}}\otimes \cH_{B_{R,k}})$ (with orthogonal subspaces
$ \cH_{B_{L,k}}\otimes \cH_{B_{R,k}}$) \cite{Hayden_2004}.
\end{itemize}

The non-negativity of quantum conditional mutual information $I(A: C|B) \geq 0$, also known as strong subadditivity of the von Neumann entropy, is a highly
non-trivial theorem
\cite{doi:10.1063/1.1666274}. It plays a central role in quantum information theory.
Recently, applications of the quantum conditional mutual information have been found in
new areas of computer science and physics, including communication and information complexity
(see \cite{10.1145/2746539.2746613,doi:10.1137/16M1061400,10.1145/3055399.3055401} and references therein), de Finetti type theorems \cite{10.1145/2488608.2488719,10.1145/2488608.2488718} and also the study of quantum many-body
systems \cite{kim2013}.

It is natural to study the robustness of the quantum Markov Chains. Many efforts have been devoted to exploring the robustness of quantum Markov Chain in terms of quantum conditional mutual information $I(A: C|B)$. Unlike the classical probability distribution, there is a tripartite quantum state far from the quantum Markov Chain while having small conditional mutual information \cite{Ibinson_2007}.
A recent breakthrough from Fawzi and Renner
shows small conditional mutual information implies high fidelity recovery from bipartite reductions
\cite{Fawzi_2015}. On the other hand, a large conditional mutual information does not imply bad recovery \cite{Sutter_2018}. The relation between quantum conditional mutual information and fidelity of recovery has received lots of attention \cite{Brand_o_2015,Li_2018,Berta_2016,Junge_2018}.

\subsection{Our results}

One of the fundamental questions regarding quantum Markov Chain is then the following:

\vspace{1mm}\emph{For a quantum Markov Chain $\rho_{ABC}$ of system $\cH_{A}\ten \cH_{B}\ten\cH_{C} $, how many copies $n$ of $\rho_{ABC}$ are necessary and sufficient to output an estimate $\hat\rho$ with
expected trace distance $\epsilon$ (or infidelity $\delta$) to the true state $\rho_{ABC}$? }

\vspace{1mm}
Consider a simple quantum Markov Chain $\rho_{ABC}=\rho_{AB}\ten I_C/d_C$. To obtain a good estimation of $\rho_{ABC}$, one must have a good estimation of $\rho_{AB}$. Similarly, one must get a good estimation of $\rho_{BC}$ by considering quantum Markov Chains of the form $\rho_{ABC}=I_A/d_A\ten \rho_{BC}$. Intuitively, a quantum Markov Chain of system $\cH_{A}\ten \cH_{B}\ten\cH_{C} $ is completely determined by its marginal of subsystems $\cH_{A}\ten \cH_{B} $ and system $\cH_{B}\ten\cH_{C}$ via the Petz map
\begin{align}
\rho_{BC}^{1/2}(\rho_B^{-1/2}\rho_{AB}\rho_B^{-1/2}\otimes id_C)\rho_{BC}^{1/2}.
\end{align}
(One can switch the role of the system $AB$ and $BC$ here.)
If we can learn $\rho_{AB}$ and $\rho_{BC}$ \emph{exactly}, we can \emph{fully} recover $\rho_{ABC}$.
It is reasonable to conjecture that $\Theta(d_A^2d_B^2+d_B^2d_C^2)$ copies are necessary and sufficient.

However, with finite copies of $\rho$, one can not learn $\rho$ exactly, but only up to some precision. On the other hand, the Petz map involves $\rho^{-1/2}$, which is unbounded and not Lipschitz continuous for $\rho$ in general. This obstacle motivates us to study the robustness of the Petz map concerning $\rho_{AB}$ and $\rho_{BC}$.

The accuracy is usually measured in the following two forms.
The fidelity of two quantum states $\rho,\sigma$ is
$F(\rho,\sigma) := \tr\sqrt{\sqrt\rho ~\sigma \sqrt\rho}$, then
the ``infidelity'' is $1-F$, represented by $\delta$, and
their trace distance is $\frac 1 2 \|\rho-\sigma\|_1$, represented by $\epsilon$.
These are related by~\cite{Nielsen:2011:QCQ:1972505}.
\begin{align}
\delta \leq \epsilon \leq \sqrt{1-(1-\delta)^2}=\sqrt{2\delta-\delta^2}.
\label{eq:FuchsGraaf}
\end{align}

Our results on the robustness of the Petz map enable us to confirm the correctness of the above intuition. More precisely, we show
\begin{thm} \label{main}
The sample complexity of quantum Markov Chain tomography of system $\cH_{A}\ten \cH_{B}\ten\cH_{C} $ is
$\tilde{\Theta}(\frac{(d_A^2+d_C^2)d_B^2}{\delta})$ in term of infidelity $\delta$, and $\tilde{\Theta}(\frac{(d_A^2+d_C^2)d_B^2}{\epsilon^2}) $ in term of trace distance $\epsilon$.
\end{thm}

By employing the robustness results, we can settle the sample complexity of the quantum Markov Chain certification problem.

\vspace{1mm}\emph{We have two quantum Markov Chains $\sigma$ and $\rho$ at $\cH_{A}\ten \cH_{B}\ten\cH_{C} $ where $\sigma$ is known. The quantum Markov Chain certification problem asks how many copies are sufficient and necessary to distinguish the case of $\rho=\sigma$ and the case where the infidelity (or trace distance) between $\rho$ and $\sigma$ is at least $\delta$ (or $\epsilon$)?
}
\vspace{1mm}

More precisely, we show
\begin{thm} \label{main-2}
The sample complexity of quantum Markov Chain certification is ${\Theta}(\frac{(d_A+d_C)d_B}{\delta})$ for indfidelity $\delta$, and ${\Theta}(\frac{(d_A+d_C)d_B}{\epsilon^2})$ for trace distance $\epsilon$.
\end{thm}

Another application of the robustness results for the Petz map is the quantum Markov Chain certification problem.

\vspace{1mm}\emph{For tripartite quantum state $\rho_{ABC}$ of system $\cH_{A}\ten \cH_{B}\ten\cH_{C} $, how many copies are sufficient and necessary to distinguish the following two cases: $\rho_{ABC}$ is a quantum Markov Chain, and $\rho_{ABC}$ is $\epsilon$-far from its Petz recovery state $\rho_{BC}^{1/2}(\rho_B^{-1/2}\rho_{AB}\rho_B^{-1/2}\otimes I_C)\rho_{BC}^{1/2}$.
}
\vspace{1mm}

Here, we choose the Petz recovery state but not general quantum Markov Chains because of the complicated structure of quantum Markov Chains. According to the literature, the closeness to a recovery state is more suitable than the distance to quantum Markov Chains, and the Petz map provides the most natural recovery state. We choose the closeness to a recovery state because bad recovery implies large conditional mutual information, but the converse is not valid in general \cite{Sutter_2018}.

In the classical distribution version, the Petz recovery state is always a Markov Chain, i.e. distribution with $0$ conditional mutual information. \cite{Canonne:2018:TCI:3188745.3188756} gives the first conditional independence tester with {\em sublinear} sample complexity. The method is used in \cite{yu2019quantum} to provide a subquadratic tester for the quantum Markov Chain when $B$ is classical.

The following result provides a tester for general quantum Markov Chains.
\begin{thm} \label{main-3}
$\tilde{O}(\frac{\min\{d_Ad_B^3d_C^3,d_A^3d_B^3d_C\}}{\epsilon^2})$ copies are sufficient to certify whether
$\rho_{ABC}$ is a quantum Markov Chain or $\epsilon$-far from its Petz recovery state for trace distance $\epsilon$.
\end{thm}
For quantum systems with $d_A\gg d_Bd_C$ or $d_C\gg d_Ad_B$, this is better than the standard full state tomography which uses $\tilde{\Theta}(\frac{d_A^2d_B^2d_C^2}{\epsilon^2})$ copies.

We generalized the tomography results into multipartite quantum Markov Chain of the system $\otimes_{i=1}^n \cH_i$. Here $\rho$ is called a quantum Markov Chain in $\otimes_{i=1}^n \cH_i$ if
\begin{align}
\rho=\mathcal{N}_{n-1}\circ\cdots\mathcal{N}_2 (\rho_{1,2})
\end{align}
where $\mathcal{N}_{i}: \cB(\cH_{i})\mapsto \cB(\cH_{i}\otimes\cH_{i+1})$ with form $\mathcal{N}_{i}(X)=\rho_{i,i+1}^{1/2}(\rho_i^{-1/2}X\rho_i^{-1/2}\otimes id_{i+1})\rho_{i,i+1}^{1/2}$, the Petz map defined of system $\cH_{i}\otimes\cH_{i+1}$ with $\rho_{s}$ denotes the reduced density matrix of system $s\subset\{1,2,\cdots,n\}$.

\begin{thm} \label{main-4}
$\tilde{O}(\frac{n^2\max_{i} \{d_i^2d_{i+1}^2\}}{\delta})$ copies are sufficient to tomography a quantum Markov Chain in term of infidelity $\delta$ with $d_i$ be the dimension of the $i$-th subsystem.
\end{thm}

This could be exponentially better than standard tomography method of quantum system $\otimes_{i=1}^n \cH_i$ which cost $\tilde{O}(\frac{n^2\max_{i} \{\Pi_i d_i^2\}}{\delta})$ copies.

\section{Preliminaries}
\subsection{Basic quantum mechanics}
An isolated physical system is associated with a
Hilbert space, called the {\it state space}. A {\it pure state} of a
quantum system is a normalized vector in its state space, denoted by the Dirac notation $\ket{\varphi}$. A
{\it mixed state} is represented by a density operator on the state
space. Here, a density operator $\rho$ on $d$-dimensional Hilbert space $\cH$ is a
semi-definite positive linear operator such that $\tr(\rho)=1$.

The state space of a composed quantum system is the tensor product of the state spaces of its component systems.
Let $\cH_k$ be a Hilbert space with dimension $d_k$. One can define a Hilbert space $\bigotimes_{k=1}^{n}\cH_{k}$ as the tensor product of Hilbert spaces $\cH^{k}$.
The quantum state on the multipartite system $\bigotimes_{k=1}^{n}\cH_{k}$ is a semi-definite positive linear operator such that $\tr(\rho)=1$.

\subsection{Quantum measurement}

A positive-operator valued measure (POVM) is
described by a collection of matrices $\{M_i\}$ with $M_i\geq 0$ and
\begin{align*}\sum_{i}M_i=I_{d}.\end{align*}
If the state of a quantum system was $\rho$
immediately before measurement $\{M_i\}$ was performed,
the probability of that result $i$ recurring is
\begin{align*}p(i)=\tr(M_i\rho).\end{align*}

\subsection{Quantum operation}
Quantum operations describe a broad class of transformations a quantum mechanical system can undergo. A quantum operation ${\Phi}$ can be represented in an elegant form known as the operator-sum
representation
\begin{align*}
{\Phi}(\rho)=\sum_{i} A_i\rho A_i^{\dag}
\end{align*} with $\sum_{i} A_i^{\dag}A_i=I$. The adjoint map $\Phi^*$ of $\Phi$ is defined as
\begin{align*}
\Phi^*(X)=\sum_i A_i^{\dag} X A_i.
\end{align*}
$\Phi^*$ is a quantum operation iff $\sum_{i} A_i^{\dag}A_i=\sum_{i}A_i A_i^{\dag}=I$.

Any quantum operation $\Phi$ applied on system $\cH$ has a Stinespring Dilation, i.e., there exists an isometry $V:\cH\mapsto \cH\ten\cH_E$ such that
\begin{align*}
\Phi(\rho)=\tr_E(V\rho V^{\dag}),
\end{align*}
where $V$ is an isometry if $V^{\dag}V=I_{\cH}$.

\subsection{Norms and fidelity}
Recall that the Schatten $p$-norm of a general operator $A$ is defined as
\begin{align*}
||A||_p=(\mathrm{Tr}|A|^p)^{1/p}\ ,\ 1<p<\infty,
\end{align*}
where $|A|\equiv\sqrt{A^\dag A}$ is the positive square root of $A^\dag A$. We denote $S_p(\cH)$ for the Schatten $p$-class operator on the Hilbert space $\cH$. For $p=1$, $\norm{\cdot }{1}$ is the trace class norm and the case $p=2$ is Hilbert-Schmidt norm. We will also denote $\norm{A}{\infty}$ as the operator norm of $A$.

The $\ell_1$ distance (also called trace distance) between quantum states $\rho$ and $\sigma$ is defined as
\begin{align*}
||\rho-\sigma||_1\equiv\mathrm{Tr}|\rho-\sigma|
\end{align*} and their $\ell_2$ distance is defined as
\begin{align*}
||\rho-\sigma||_2=\sqrt{\mathrm{Tr}(\rho-\sigma)^2}.
\end{align*}
For $\rho,\sigma\in\mathcal{D}(\C^{d})$, we have the following relation between $\ell_1$ and $\ell_2$ distances,
\begin{align*}
||\rho-\sigma||_2\leq ||\rho-\sigma||_1\leq \sqrt{d}||\rho-\sigma||_2.
\end{align*}

We will also use fidelity $F(\cdot,\cdot)$ to measure the "closeness" of two quantum states:
\begin{align*}
F(\rho,\sigma)=\tr|\rho^{\frac12}\sigma^{\frac12}|.
\end{align*}

It is known that the following relation between fidelity and $\ell_1$ distance hold.
\begin{lemma} \cite{Audenaert_2008,OW17}
For density matrices $\rho$ and $\sigma$, we have
\begin{align*}
&F(\rho,\sigma)+\frac{ \norm{\rho-\sigma}{1}}{2}\geq 1,\\
&F^2(\rho,\sigma)+\frac{ \norm{\rho-\sigma}{1}^2}{4}\geq 1, \\
&F^2(\rho,\sigma)\leq \tr(\rho^{\frac12}\sigma^{\frac12})\le F(\rho,\sigma).
\end{align*}
\end{lemma}

\subsection{Previous results}
We will use the following results on quantum property testing.
\begin{lemma} [\cite{HHJ+16}]\label{HHJ}
There exists a quantum measurement $\{ M(\sigma) \}$, such that for any $d$-dimensional quantum state $\rho$, the measurement on $\rho^{\otimes n}$
outputs an estimate $\hat \rho $
such that $F(\hat \rho,\rho) \ge 1 - \delta$
with probability at least
\begin{align*}
1 - (n+1)^{3d^2} e^{-2 n \delta},
\end{align*}
which is at least $1-e^{-kd^2\ln(d/\delta)}$ for $n \geq \frac{100 kd^2 \ln (d/\delta)}{\delta}$.

On the other hand, to achieve trace distance error $\epsilon$ between $\hat \rho$ and $\rho$ with a probability of at least $2/3$,
\begin{align*}
n\geq \Omega(d^2/\epsilon^2)
\end{align*}

According to the relation between $\epsilon$ and $\delta$, we have
\begin{align*}
n\geq \Omega(dr/\delta),
\end{align*}
to achieve infidelity $\delta$.
\end{lemma}

\begin{lemma}[\cite{OW16}]\label{OW16}
$O(\frac{d^2}{\epsilon^2})$ copies of $\rho$ is enough to output an estimate $\hat \rho $
such that $||\hat \rho-\rho||_1<\epsilon$ with probability $2/3$.
\end{lemma}

\begin{lemma}[\cite{OW15,BOW17}]\label{BOW17}
The sample complexity of certifying whether unknown quantum state $\rho$ equals the known state $\sigma$ or $\epsilon$-far (resp. $\delta$-far) from $\sigma$ is $\Theta(\frac{d}{\epsilon^2})$ (resp. $\Theta(\frac{d}{\delta})$ ). The statement is still true if $\sigma$ is maximally mixed.

The sample complexity of certifying whether two unknown quantum states $\rho$ and $\sigma$ satisfies $||\rho-\sigma||_2<0.99\epsilon$ or $||\rho-\sigma||_2>\epsilon$ is $\Theta(\frac{1}{\epsilon^2})$.
\end{lemma}

\subsection{Petz map for general quantum operation}
The relative entropy between quantum state $\rho$ and $\sigma$ is
\begin{align*}
S(\rho||\sigma)=\tr(\rho\log \rho)-\tr(\rho\log\sigma)
\end{align*}
if the support of $\rho$ is contained in the support of $\sigma$, and
$\infty$ otherwise.

The strong subadditivity is equivalent to the monotonicity of relative entropy
under a partial trace operation:
\begin{align*}
S(\rho_{ABC}||\rho_A\otimes\rho_{BC})-S(\rho_{AB}||\rho_A\otimes\rho_B)=S(AB)_{\rho}+S(BC)_{\rho}-S(B)_{\rho}-S(ABC)_{\rho}=I(A:C|B).
\end{align*}

Petz considered the general Petz map for quantum operation $\mathcal{N}$.
\begin{lemma}[Petz~\cite{Petz_1996}]
\label{thm:petz}
For states $\rho$ and $\sigma$ and quantum operation $\Phi$,
$$S(\rho\|\sigma) = S(\Phi(\rho)\|\Phi(\sigma))$$
iff there exists a quantum operation $\widehat{\Phi}$
such that
$$\widehat{\Phi}(\Phi(\rho))=\rho,\quad \widehat{\Phi}(\Phi(\sigma))=\sigma.$$
Furthermore, on the support of $\Phi(\sigma)$,
\begin{equation}
\label{eq:transpose:channel}
\widehat{\Phi}(\alpha)=\sigma^{\frac{1}{2}}
\Phi^*\left(
\Phi(\sigma)^{-\frac{1}{2}}\alpha \Phi(\sigma)^{-\frac{1}{2}}
\right)\sigma^{\frac{1}{2}}.
\end{equation}
Here $\Phi^*$ denotes the adjoint map of $\Phi$.
\end{lemma}

For tripartite quantum system $\cH_{A}\ten \cH_{B}\ten\cH_{C} $, we choose $\Phi: \cH_{BC}\mapsto \cH_B$ to be the partial trace operator, i.e., for any $\sigma_{BC}$
\begin{align*}
\Phi(\sigma_{BC})=\sigma_{BC}.
\end{align*}
Then
$I(A:C|B)_{\rho}=0$ iff
\begin{align*}
S(\rho_{ABC}||\rho_A\otimes\rho_{BC})=S(\rho_{AB}||\rho_A\otimes\rho_B)
\end{align*}
Lemma \ref{thm:petz} implies
\begin{align}\label{eq:petz}
\rho_{ABC}=\rho_{BC}^{\frac12}\rho_{B}^{-\frac12}(\rho_{AB}\ten id_C)\rho_{B}^{-\frac12}\rho_{BC}^{\frac12}\pl .
\end{align}
One can verify that
\begin{lemma}
For any $\rho_{BC}$, the
Petz map $\mathcal{N}_{B\mapsto BC}(X)=\rho_{BC}^{\frac12}\rho_{B}^{-\frac12}(X\ten id_C)\rho_{B}^{-\frac12}\rho_{BC}^{\frac12}$ is a quantum operation.
\end{lemma}
\begin{proof}
Let $M_i=\rho_{BC}^{\frac12}\rho_{B}^{-\frac12} \ten \ket{i}_C$, one can verify that
\begin{align*}
\sum_i M_i^{\dag} X M_i=\rho_{BC}^{\frac12}\rho_{B}^{-\frac12}(X\ten \sum_{i}\op{i}{i})\rho_{B}^{-\frac12}\rho_{BC}^{\frac12}=\rho_{BC}^{\frac12}\rho_{B}^{-\frac12}(X\ten id_C)\rho_{B}^{-\frac12}\rho_{BC}^{\frac12}=\mathcal{N}_{B\mapsto BC}(X).
\end{align*}
On the other hand,
\begin{align*}
\sum_{i} M_i^{\dag}M_i=\sum_{i} \rangle i|\rho_{B}^{-\frac12} \rho_{BC}\rho_{B}^{-\frac12}\ket{i}=\rho_{B}^{-\frac12} \sum_{i} \rangle i|\rho_{BC}\ket{i}\rho_{B}^{-\frac12}=\rho_{B}^{-\frac12}\rho_{B} \rho_{B}^{-\frac12}=I_{\cH_B}.
\end{align*}
\end{proof}
\section{Continuity of the Petz map}

In this section, we study the continuity of the Petz map.

\subsection{Useful Lemmata}
We start with the key estimate in our results.
\begin{lemma}\label{distance}
For density matrices, $\rho_{AB},\sigma_{AB}$, $\rho_{BC},\sigma_{BC}$ and $\rho_{B},\sigma_B$, with
\begin{align*}
\rho_B=\tr_A\rho_{AB},\\
\sigma_B=\tr_C\sigma_{BC},
\end{align*}
we have
\[
\norm{\rho_{BC}^{\frac12}\rho_{B}^{-\frac12}\rho_{AB}^{\frac12}
-\sigma_{BC}^{\frac12}\sigma_{B}^{-\frac12}\sigma_{AB}^{\frac12}}{2}\le \norm{\rho_{AB}^{\frac12}-\sigma_{AB}^{\frac12}}{2}+\norm{\rho_{BC}^{\frac12}-\sigma_{BC}^{\frac12}}{2}+\norm{\rho_{B}^{\frac12}-\sigma_{B}^{\frac12}}{2}.
\]
\end{lemma}
\begin{proof}
By the triangle inequality,
\begin{align*}
&\norm{\rho_{BC}^{\frac12}\rho_{B}^{-\frac12}\rho_{AB}^{\frac12}
-\sigma_{BC}^{\frac12}\sigma_{B}^{-\frac12}\sigma_{AB}^{\frac12}}{2}\\
\le &\norm{\rho_{BC}^{\frac12}\rho_{B}^{-\frac12}\rho_{AB}^{\frac12}
-\sigma_{BC}^{\frac12}\rho_{B}^{-\frac12}\rho_{AB}^{\frac12}}{2}+
\norm{\sigma_{BC}^{\frac12}\rho_{B}^{-\frac12}\rho_{AB}^{\frac12}
-\sigma_{BC}^{\frac12}\sigma_{B}^{-\frac12}\rho_{AB}^{\frac12}}{2}
+\norm{\sigma_{BC}^{\frac12}\sigma_{B}^{-\frac12}\rho_{AB}^{\frac12}
-\sigma_{BC}^{\frac12}\sigma_{B}^{-\frac12}\sigma_{AB}^{\frac12}}{2}
\end{align*}
For the first term,
\begin{align*}
&\norm{\rho_{BC}^{\frac12}\rho_{B}^{-\frac12}\rho_{AB}^{\frac12}
-\sigma_{BC}^{\frac12}\rho_{B}^{-\frac12}\rho_{AB}^{\frac12}}{2}^2\\
=&\tr_{ABC}( (\rho_{BC}^{\frac12}-\sigma_{BC}^{\frac12})^2\rho_{B}^{-\frac12}\rho_{AB}\rho_{B}^{-\frac12})\\
=&\tr_{BC}( (\rho_{BC}^{\frac12}-\sigma_{BC}^{\frac12})^2)\\
=&\norm{\rho_{BC}^{\frac12}-\sigma_{BC}^{\frac12}}{2}^2
\end{align*}
Similarly, for the third term
\begin{align*}
&\norm{\sigma_{BC}^{\frac12}\sigma_{B}^{-\frac12}\rho_{AB}^{\frac12}
-\sigma_{BC}^{\frac12}\sigma_{B}^{-\frac12}\sigma_{AB}^{\frac12}}{2}^2\\
=&\tr_{ABC}(\sigma_{B}^{-\frac12}\sigma_{BC}\sigma_{B}^{-\frac12} (\rho_{AB}^{\frac12}-\sigma_{AB}^{\frac12})^2)\\
=&\tr_{ABC}( (\rho_{AB}^{\frac12}-\sigma_{AB}^{\frac12})^2)\\
= &\norm{\rho_{AB}^{\frac12}-\sigma_{AB}^{\frac12}}{2}^2
\end{align*}
For the second term
\begin{align*}
&\norm{\sigma_{BC}^{\frac12}\rho_{B}^{-\frac12}\rho_{AB}^{\frac12}
-\sigma_{BC}^{\frac12}\sigma_{B}^{-\frac12}\rho_{AB}^{\frac12}}{2}^2\\
=&\tr_{ABC}( \rho_{AB}\rho_{B}^{-\frac12}\sigma_{BC}\rho_{B}^{-\frac12})+\tr_{ABC}( \rho_{AB}\sigma_{B}^{-\frac12}\sigma_{BC}\sigma_{B}^{-\frac12})-\tr_{ABC}( \rho_{AB}\sigma_{B}^{-\frac12}\sigma_{BC}\rho_{B}^{-\frac12})-\tr_{ABC}( \rho_{AB}\rho_{B}^{-\frac12}\sigma_{BC}\sigma_{B}^{-\frac12})\\
=&\tr_{B}( \rho_{B}\rho_{B}^{-\frac12}\sigma_{B}\rho_{B}^{-\frac12})+\tr_{B}( \rho_{B}\sigma_{B}^{-\frac12}\sigma_{B}\sigma_{B}^{-\frac12})-\tr_{B}( \rho_{B}\sigma_{B}^{-\frac12}\sigma_{B}\rho_{B}^{-\frac12})-\tr_{B}( \rho_{B}\rho_{B}^{-\frac12}\sigma_{B}\sigma_{B}^{-\frac12})\\
=&2-2\tr( \sigma_{B}^{\frac12}\rho_{B}^{\frac12})\\
=& \tr((\rho_B^{\frac12}-\sigma_B^{\frac12})^2)\\
=& \norm{\rho_{B}^{\frac12}-\sigma_{B}^{\frac12}}{2}^2
\end{align*}
This completes the proof.
\end{proof}

\begin{rem}{\rm In the statement of the above lemma, we do not require
\begin{align*}
\rho_B=\tr_A\rho_{AB}=\tr_C\rho_{BC},\\
\sigma_B=\tr_C\sigma_{BC}=\tr_A\sigma_{AB}.
\end{align*}
By symmetry, an equivalent condition is that
\begin{align*}
\rho_B=\tr_C\rho_{BC},\\
\sigma_B=\tr_A\sigma_{AB}.
\end{align*}
For example, one can have a simplified case
\begin{align*}
\norm{\rho_{BC}^{\frac12}\rho_{B}^{-\frac12}\rho_{AB}^{\frac12}
-\sigma_{BC}^{\frac12}\sigma_{B}^{-\frac12}\rho_{AB}^{\frac12}}{2}\le \norm{\rho_{BC}^{\frac12}-\sigma_{BC}^{\frac12}}{2}+\norm{\rho_{B}^{\frac12}-\sigma_{B}^{\frac12}}{2}.
\end{align*}
}
\end{rem}

We use Lemma \ref{distance} will derive the continuity of the Petz map in terms of $\epsilon$ and $\delta$. Before that, we describe the relation between different distances. The following lemma shows that the square root of infidelity $\sqrt{1-F}$ is equivalent to $\norm{\rho^{\frac12}-\sigma^{\frac12}}{2}$ up to a constant $\sqrt{2}$.
\begin{lemma}\label{f-1}
For density matrices $\rho$ and $\sigma$, we have
\[
\sqrt{2}\sqrt{1-F(\rho,\sigma)}\le\norm{\rho^{\frac12}-\sigma^{\frac12}}{2}\le 2\sqrt{1-F(\rho,\sigma)}.
\]
\end{lemma}
\begin{proof}
We observe that
\[
\norm{\rho^{\frac12}-\sigma^{\frac12}}{2}^2=2-2\tr(\rho^{\frac12}\sigma^{\frac12}).
\]
On the one hand,
\[
2-2\tr(\rho^{\frac12}\sigma^{\frac12}) \le 2-2 F^2(\rho,\sigma) \le 4(1-F(\rho,\sigma))
\]
On the other hand,
\begin{align*}
2-2\tr(\rho^{\frac12}\sigma^{\frac12}) \ge 2-2 F(\rho,\sigma) =2 (1-F(\rho,\sigma)).
\end{align*}
\end{proof}

\begin{lemma}\label{fidelity}
For operators $X,Y$ with $||X||_2=||Y||_2=1$ and
\[\norm{X-Y}{2}\le \delta\]
then,
\[
F(X^{\dag}X,Y^{\dag}Y)\ge 1-\delta^2/2.
\]
\end{lemma}
\begin{proof}
Using polar decomposition, there exist unitaries $U$ and $V$ such that $X=U|X|$ and $Y=V|Y|$. We let $A=|X||Y|$. Denote $\mathbf{Re}(z)$ as the real part of a complex number $z$.
We observe that
\begin{align*}
&\norm{X-Y}{2}^2=\tr X^{\dag}X+\tr Y^{\dag}Y-\tr(X^{\dag}Y+XY^{\dag})=2-2 \mathbf{Re}(\tr XY^{\dag})\le \delta^2\\
\Longleftrightarrow & \mathbf{Re}(\tr XY^{\dag})\geq 1-\delta^2/2\\
\Longrightarrow & |\tr XY^{\dag}|\geq 1-\delta^2/2\\
\Longrightarrow & |\tr U |X| |Y|V^{\dag} |\geq 1-\delta^2/2\\
\Longleftrightarrow & |\tr UAV^{\dag}U |\geq 1-\delta^2/2\\
\Longrightarrow & \norm{A}{1} \geq 1-\delta^2/2,
\end{align*}
where in the last step, we use H\"older. On the other hand,
\[
F(X^{\dag}X,Y^{\dag}Y)=\tr |\sqrt{X^{\dag}X}\sqrt{Y^{\dag}Y}|=\tr | \ |X| |Y|\ |=\tr|A|
\]
Therefore,
\[
F(X^{\dag}X,Y^{\dag}Y)=\norm{A}{1}\ge 1-\delta^2/2.
\]
\end{proof}

The following lemma provides a relationship between $\ell_1$ norm and $\norm{\rho^{\frac12}-\sigma^{\frac12}}{2}$.
\begin{lemma}\label{Bhatia}
Let $\rho,\sigma$ be two positive operators, then we have
\[\norm{\rho^{\frac12}-\sigma^{\frac12}}{2}^2\le \norm{\rho-\sigma}{1} \le 2\norm{\rho^{\frac12}-\sigma^{\frac12}}{2}\]
\end{lemma}
\begin{proof}The lower estimate is a special case of Theorem X.1.6 of \cite{bhatia2013matrix} and the upper estimate is Lemma 2.2 \cite{carlen2020recovery}.
\end{proof}

For the Schatten $p$-norm, we have
\begin{lemma}\label{lemma:4to2}
Let $X,Y$ be operators with $\norm{X}{2p}=1$ and $\norm{Y}{2p}=1$.
\[\norm{|X|-|Y|}{2p}^2\le\norm{X^{\dag}X-Y^{\dag}Y}{p}\le 2\norm{X-Y}{2p}\]
\end{lemma}
\begin{proof}Using H\"older inequality,
\begin{align*}
\norm{X^{\dag}X-Y^{\dag}Y}{p}\le& \norm{X^{\dag}X-Y^{\dag}X}{p}+\norm{Y^{\dag}X-Y^{\dag}Y}{p}\\
\le& \norm{X^{\dag}-Y^{\dag}}{2p}
\norm{X}{2p}+\norm{Y^{\dag}}{2p}\norm{X-Y}{2p}=2\norm{X-Y}{2p}\pl.
\end{align*}
Another inequality is, again, from Theorem X.1.6 of \cite{bhatia2013matrix}
\end{proof}

\subsection{Continuity of the Petz map in terms of infidelity}
We have the following continuity of the Petz map in terms of infidelity.
\begin{theorem}\label{petz:fidelity}
For density matrices, $\rho_{AB},\sigma_{AB}$, $\rho_{BC},\sigma_{BC}$ and $\rho_{B},\sigma_B$, with
\begin{align*}
\rho_B=\tr_A\rho_{AB},\\
\sigma_B=\tr_C\sigma_{BC},
\end{align*}
If \[F(\rho_{AB},\sigma_{AB})\ge 1-\delta_1\pl, F(\rho_{BC},\sigma_{BC})\ge 1-\delta_2\pl, F(\rho_{B},\sigma_{B})\ge 1-\delta_3,\]
we have
\[
F(\rho_{BC}^{\frac12}\rho_{B}^{-\frac12}\rho_{AB}\rho_{B}^{-\frac12}\rho_{BC}^{\frac12},\sigma_{BC}^{\frac12}\sigma_{B}^{-\frac12}\sigma_{AB}\sigma_{B}^{-\frac12}\sigma_{BC}^{\frac12})\geq 1-2(\delta_1^{\frac12}+\delta_2^{\frac12}+\delta_3^{\frac12})^2.
\]
\end{theorem}
\begin{proof}
According to Lemma \ref{f-1}, we have
\begin{align*}
\norm{\rho_{AB}^{\frac12}-\sigma_{AB}^{\frac12}}{2}\le 2{\delta_1}^{\frac12}, \norm{\rho_{BC}^{\frac12}-\sigma_{BC}^{\frac12}}{2}\le 2{\delta_2}^{\frac12},
\norm{\rho_{B}^{\frac12}-\sigma_{B}^{\frac12}}{2}\le 2{\delta_3}^{\frac12}.
\end{align*}
Using Lemma \ref{distance}, we have
\begin{align*}
\norm{\rho_{BC}^{\frac12}\rho_{B}^{-\frac12}\rho_{AB}^{\frac12}
-\sigma_{BC}^{\frac12}\sigma_{B}^{-\frac12}\sigma_{AB}^{\frac12}}{2}\le2(\delta_1^{\frac12}+\delta_2^{\frac12}+\delta_3^{\frac12}).
\end{align*}
Lemma \ref{fidelity} implies
\begin{align*}
F(\rho_{BC}^{\frac12}\rho_{B}^{-\frac12}\rho_{AB}\rho_{B}^{-\frac12}\rho_{BC}^{\frac12},\sigma_{BC}^{\frac12}\sigma_{B}^{-\frac12}\sigma_{AB}\sigma_{B}^{-\frac12}\sigma_{BC}^{\frac12})\geq 1-2(\delta_1^{\frac12}+\delta_2^{\frac12}+\delta_3^{\frac12})^2.
\end{align*}
\end{proof}
We immediately have the following result,
\begin{cor}\label{qmc:fidelity}
Let $\rho_{ABC},\sigma_{ABC}$ be two quantum Markov Chains with \[F(\rho_{AB},\sigma_{AB})\ge 1-\delta_1\pl, F(\rho_{BC},\sigma_{BC})\ge 1-\delta_2\pl, F(\rho_{B},\sigma_{B})\ge 1-\delta_3\]
then
\[
F(\rho_{ABC},\sigma_{ABC})\geq 1-2(\delta_1^{\frac12}+\delta_2^{\frac12}+\delta_3^{\frac12})^2.
\]
\end{cor}

\subsection{Continuity of the Petz map in terms of trace distance}

Below is the continuity of the Petz map in terms of $\ell_1$-norm.

\begin{theorem}\label{petz:distance}
For density matrices, $\rho_{AB},\sigma_{AB}$, $\rho_{BC},\sigma_{BC}$ and $\rho_{B},\sigma_B$, with
\begin{align*}
\rho_B=\tr_A\rho_{AB},\\
\sigma_B=\tr_C\sigma_{BC},
\end{align*}
If \[\norm{\rho_{AB}-\sigma_{AB}}{1}\le \eps_1\pl, \norm{\rho_{BC}-\sigma_{BC}}{1}\le \eps_2\pl, \norm{\rho_{B}-\sigma_{B}}{1}\le \eps_3,\]
we have
\[
||\rho_{BC}^{\frac12}\rho_{B}^{-\frac12}\rho_{AB}\rho_{B}^{-\frac12}\rho_{BC}^{\frac12}-\sigma_{BC}^{\frac12}\sigma_{B}^{-\frac12}\sigma_{AB}\sigma_{B}^{-\frac12}\sigma_{BC}^{\frac12}||_1\leq 2(\eps_1^{\frac12}+\eps_2^{\frac12}+\eps_3^{\frac12}).
\]
\end{theorem}
\begin{proof}
According to Lemma \ref{Bhatia}, we have
\begin{align*} \norm{\rho_{AB}^{\frac12}-\sigma_{AB}^{\frac12}}{2}\le \eps_1^{\frac12}, \norm{\rho_{BC}^{\frac12}-\sigma_{BC}^{\frac12}}{2}\le \eps_2^{\frac12},
\norm{\rho_{B}^{\frac12}-\sigma_{B}^{\frac12}}{2}\le \eps_3^{\frac12}.
\end{align*}
Using Lemma \ref{fidelity}, we obtain
\begin{align*}
||\rho_{BC}^{\frac12}\rho_{B}^{-\frac12}\rho_{AB}\rho_{B}^{-\frac12}\rho_{BC}^{\frac12}-\sigma_{BC}^{\frac12}\sigma_{B}^{-\frac12}\sigma_{AB}\sigma_{B}^{-\frac12}\sigma_{BC}^{\frac12}||_1\le &
2\norm{\rho_{BC}^{\frac12}\rho_{B}^{-\frac12}\rho_{AB}^{\frac12}
-\sigma_{BC}^{\frac12}\sigma_{B}^{-\frac12}\sigma_{AB}^{\frac12}}{2}\\
\le & 2(\eps_1^{\frac12}+\eps_2^{\frac12}+\eps_3^{\frac12})\pl,
\end{align*}
\end{proof}
\begin{cor}\label{cor:L2}
Let $\rho_{ABC},\sigma_{ABC}$ be two quantum Markov Chains \[\norm{\rho_{AB}-\sigma_{AB}}{1}\le \eps_1\pl, \norm{\rho_{BC}-\sigma_{BC}}{1}\le \eps_2\pl, \norm{\rho_{B}-\sigma_{B}}{1}\le \eps_3\]
Then
\[ \norm{\rho_{ABC}-\sigma_{ABC}}{1}\le 2(\eps_1^{\frac12}+\eps_2^{\frac12}+\eps_3^{\frac12})\]
\end{cor}

\subsection{Continuity of general Petz maps}
Let $\Phi$ be a quantum channel. The Petz recovery map associated with an input state $\sigma$ is
\[\hat{\Phi}_\sigma(X)=\sigma^{1/2}\Phi^*(\Phi(\sigma)^{-1/2}X\Phi(\sigma)^{-1/2})\sigma^{1/2}\pl.\]
Let $\Phi(\rho)=\tr_E(V\rho V^*)$ be the Stinespring dilation of $\Phi$. We have
\[\hat{\Phi}_\sigma(X)=\sigma^{1/2} V^*\Big(\Phi(\sigma)^{-1/2}X\Phi(\sigma)^{-1/2}\ten I_E\Big)V\sigma^{1/2} . \]

\begin{theorem}\label{Petz}
Let $\Phi:B\to A$ be a quantum channel and $\rho_{BC},\sigma_{BC}$ be two bipartite state. Denote $\hat{\Phi}_\rho$ and $\hat{\Phi}_{\sigma}$ be the Petz recovery maps for $\rho_B,\sigma_B$ respectively. Then
for any bipartite states $\rho_{BC}$ with $\tr_C\rho_BC=\rho_B$ and $\omega_{BC}$,
\begin{align*}
\norm{\hat{\Phi}_\rho(\Phi(\rho_{BC}))-\hat{\Phi}_\sigma(\Phi(\omega_{BC}))}{1}\le & 2\norm{\rho_B^{1/2}-\sigma_B^{1/2}}{2}+2\norm{\Phi(\rho_B)^{1/2}-\Phi(\sigma_B)^{1/2}}{2}
\\ &+ 2 \norm{\Phi(\rho_{BC})^{1/2}-\Phi(\omega_{BC})^{1/2}}{2}
\end{align*}
\end{theorem}
\begin{proof}
Let $\Phi(\rho)=\tr_E(V\rho V^{\dagger})$ be the Stinespring dilation of $\Phi$. We have
\[\hat{\Phi}_\sigma(X)=\sigma_B^{1/2} V^{\dagger}\Big(\Phi(\sigma_B)^{-1/2}X\Phi(\sigma_B)^{-1/2}\ten I_E\Big)V\sigma_B^{1/2} . \]
Then, by the triangle inequality,
\begin{align*}
&\norm{\hat{\Phi}_\rho(\Phi(\rho_{BC}))-\hat{\Phi}_\sigma(\Phi(\omega_{BC}))}{1}\\
\le &2\norm{\sigma_B^{1/2} V^{\dagger}\Big(\Phi(\sigma_B)^{-1/2}\Phi(\sigma_{BC})^{1/2}\ten I_E\Big)-\rho_B^{1/2} V^{\dagger}\Big(\Phi(\rho_B)^{-1/2}\Phi(\rho_{BC})^{1/2}\ten I_E\Big)}{2}\\
\le &2\norm{(\rho_B^{1/2} -\sigma_B^{1/2}) V^{\dagger}\Big(\Phi(\rho_B)^{-1/2}\Phi(\rho_{BC})^{1/2}\ten I_E\Big)}{2}\\ +&2\norm{\sigma_B^{1/2} V^{\dagger}\Big(\big(\Phi(\sigma_B)^{-1/2}-\Phi(\rho_B)^{-1/2}\big)\Phi(\rho_{BC})^{1/2}\ten I_E\Big)}{2}
\\ +&2\norm{\sigma_B^{1/2} V^{\dagger}\Big(\Phi(\sigma_B)^{-1/2}\big(\Phi(\rho_{BC})^{1/2}-\Phi(\omega_{BC})^{1/2}\big)\ten I_E\Big)}{2}
\end{align*}
For the first term,
\begin{align*}
&\norm{(\rho_B^{1/2} -\sigma_B^{1/2}) V^{\dagger}\Big(\Phi(\rho_B)^{-1/2}\Phi(\rho_{BC})^{1/2}\ten I_E\Big)}{2}^2\\
=&\norm{(\rho_B^{1/2} -\sigma_B^{1/2}) \Phi^*\big(\Phi(\rho_B)^{-1/2}\Phi(\rho_{BC})\Phi(\rho_B)^{-1/2}\ten I_E\big)(\rho_B^{1/2} -\sigma_B^{1/2})}{1}
\\
=&\tr_{BC}\Big((\rho_B^{1/2} -\sigma_B^{1/2})^2 \Phi^*\big(\Phi(\rho_B)^{-1/2}\Phi(\rho_{BC})\Phi(\rho_B)^{-1/2}\big)\Big)
\\
=&\tr_{B}\Big((\rho_B^{1/2} -\sigma_B^{1/2})^2 \Phi^*\big(\Phi(\rho_B)^{-1/2}\Phi(\rho_{B})\Phi(\rho_B)^{-1/2}\big)\Big)
\\
=&\tr_{B}\Big((\rho_B^{1/2} -\sigma_B^{1/2})^2 \Phi^*(I)\Big)
\\
=&\tr_{B}\Big((\rho_B^{1/2} -\sigma_B^{1/2})^2\Big)=\norm{\rho_B^{1/2}-\sigma_B^{1/2}}{2}
\end{align*}
For the second term
\begin{align*}
&\norm{\sigma_B^{1/2}V^{\dagger}\Big(\big(\Phi(\sigma_B)^{-1/2}-\Phi(\rho_B)^{-1/2}\big)\Phi(\rho_{BC})^{1/2}\ten I_E\Big)}{2}^2\\
=&\tr_{BC}\left( \sigma_B \Phi^*\left( \left(\Phi(\sigma_B)^{-1/2}-\Phi(\rho_B)^{-1/2}\right) \Phi(\rho_{BC}) \left(\Phi(\sigma_B)^{-1/2}-\Phi(\rho_B)^{-1/2}\right) \right) \right) \\
=&\tr_{BC}\left( \Phi(\sigma_B) \left(\Phi(\sigma_B)^{-1/2}-\Phi(\rho_B)^{-1/2}\right) \Phi(\rho_{BC}) \left(\Phi(\sigma_B)^{-1/2}-\Phi(\rho_B)^{-1/2}\right) \right) \\
=&\tr_{B}\left( \Phi(\sigma_B) \left(\Phi(\sigma_B)^{-1/2}-\Phi(\rho_B)^{-1/2}\right) \Phi(\rho_{B}) \left(\Phi(\sigma_B)^{-1/2}-\Phi(\rho_B)^{-1/2}\right) \right) \\
=& \tr_{B}\Big((\Phi(\rho_B)^{1/2} -\Phi(\sigma_B)^{1/2})^2\Big)\\
=&\norm{\Phi(\rho_B)^{1/2}-\Phi(\sigma_B)^{1/2}}{2}
\end{align*}
For the third term
\begin{align*}
&\norm{\sigma_B^{1/2} V^{\dagger}\Phi(\sigma_B)^{-1/2}\big(\Phi(\rho_{BC})^{1/2}-\Phi(\omega_{BC})\big)^{1/2}\ten I_E\Big)}{2}^2
\\=&\tr_{BCE}\Big(\big((\Phi(\rho_{BC})^{1/2}-\Phi(\omega_{BC}))^{1/2}\Phi(\sigma_B)^{-1/2}\ten I_E\big) V\sigma_B V^{\dagger}\big(\Phi(\sigma_B)^{-1/2}\Phi(\rho_{BC})^{1/2}-\Phi(\omega_{BC})^{1/2}\ten I_E\Big)
\\=&\tr_{BC}\Big(\Phi(\sigma_B)^{-1/2} \Phi(\sigma_B)\Phi(\sigma_B)^{-1/2}(\Phi(\rho_{BC})^{1/2}-\Phi(\sigma_{BC})^{1/2})^2\Big)
\\=&\tr_{BC}(\Phi(\rho_{BC})^{1/2}-\Phi(\omega_{BC})^{1/2})^2=\norm{\Phi(\rho_{BC})^{1/2}-\Phi(\omega_{BC})^{1/2} }{2}^2
\end{align*}
This completes the proof.
\end{proof}

\subsection{Continuity of the Petz map in terms of $\ell_2$ distance}

\begin{theorem}\label{Petz-2}
For density matrices, $\rho_{AB},\sigma_{AB}$, $\rho_{BC},\sigma_{BC}$ and $\rho_{B},\sigma_B$, with
\begin{align*}
\rho_B=\tr_A\rho_{AB},\\
\sigma_B=\tr_C\sigma_{BC}.
\end{align*}
We have
\begin{align*}
&\norm{\rho_{BC}^{\frac12}\rho_{B}^{-\frac12}\rho_{AB}\rho_{B}^{-\frac12}\rho_{BC}^{\frac12}
-\sigma_{BC}^{\frac12}\sigma_{B}^{-\frac12}\sigma_{AB}\sigma_{B}^{-\frac12}\sigma_{BC}^{\frac12}}{2}\\ \le & 2d_A^{1/4}\norm{\sigma_{BC}^{\frac12}-\rho_{BC}^{\frac12}}{4}+2d_A^{1/4}d_C^{1/4}\norm{\sigma_{B}^{\frac12}-\rho_{B}^{\frac12}}{4}+2d_C^{1/4}\norm{\sigma_{AB}^{\frac12}-\rho_{AB}^{\frac12}}{4}.
\end{align*}
\end{theorem}
\begin{proof}
According to Lemma \ref{lemma:4to2}, we have
\begin{align*}
&\norm{\rho_{BC}^{\frac12}\rho_{B}^{-\frac12}\rho_{AB}\rho_{B}^{-\frac12}\rho_{BC}^{\frac12}
-\sigma_{BC}^{\frac12}\sigma_{B}^{-\frac12}\sigma_{AB}\sigma_{B}^{-\frac12}\sigma_{BC}^{\frac12}}{2}
\\ \le &
2\norm{\rho_{BC}^{\frac12}\rho_{B}^{-\frac12}\rho_{AB}^{1/2}
-\sigma_{BC}^{\frac12}\sigma_{B}^{-\frac12}\sigma_{AB}^{1/2}}{4}
\\ \le & 2\norm{\rho_{BC}^{\frac12}\rho_{B}^{-\frac12}\rho_{AB}^{1/2}
-\sigma_{BC}^{\frac12}\rho_{B}^{-\frac12}\rho_{AB}^{1/2}}{4}+2\norm{\sigma_{BC}^{\frac12}(\sigma_{B}^{-\frac12}-\rho_{B}^{-\frac12})\rho_{AB}^{1/2}}{4}+2\norm{\sigma_{BC}^{\frac12}\sigma_{B}^{-\frac12}(\sigma_{AB}^{1/2}-\rho_{AB}^{1/2})}{4}
\end{align*}

For the first term,
\begin{align*}
&\norm{(\rho_{BC}^{\frac12}
-\sigma_{BC}^{\frac12})\rho_{B}^{-\frac12}\rho_{AB}^{1/2}}{4}
\\ \le &\norm{(\rho_{BC}^{\frac12}
-\sigma_{BC}^{\frac12})\rho_{B}^{-\frac12}\rho_{AB}\rho_{B}^{-\frac12} (\rho_{BC}^{\frac12}
-\sigma_{BC}^{\frac12})}{2}^{1/2}
\\ \le & d_A^{1/4} \norm{(\rho_{BC}^{\frac12}
-\sigma_{BC}^{\frac12})\rho_{B}^{-\frac12}\rho_{B}\rho_{B}^{-\frac12} (\rho_{BC}^{\frac12}
-\sigma_{BC}^{\frac12})}{2}^{1/2}
\\ = & d_A^{1/4} \norm{(\rho_{BC}^{\frac12}
-\sigma_{BC}^{\frac12})^2}{2}^{1/2}
\\ = & d_A^{1/4} \norm{\rho_{BC}^{\frac12}
-\sigma_{BC}^{\frac12}}{4}
\end{align*}
Here the first inequality uses Lemma \ref{lemma:4to2}. The second inequality follows from the following fact: For any $X_{AR}$,
\[||\tr_A X_{AR}||_2\leq d_A^{1/2}||X_{AR}||_2\ .\]
Indeed, we write
\begin{align*}
X_{AR}=\sum_{i,j}\op{i}{j}\ten X_{i,j}
\end{align*}
Then
\begin{align*}
||\tr_A X_{AR}||_2^2=||\sum_{i} X_{i,i}||_2^2\leq (\sum_{i} ||X_{i,i}||_2)^2\leq d_A (\sum_{i} ||X_{i,i}||_2^2)\leq d_A (\sum_{i,j} ||X_{i,j}||_2^2)=d_A ||X_{AR}||_2^2.
\end{align*}
The third term is similar,
\begin{align*}
&\norm{\sigma_{BC}^{\frac12}\sigma_{B}^{-\frac12}(\sigma_{AB}^{1/2}-\rho_{AB}^{1/2})}{4}
\\ \le & \norm{(\sigma_{AB}-\rho_{AB}^{1/2})\sigma_{B}^{-\frac12}\sigma_{BC}\sigma_{B}^{-\frac12}(\sigma_{AB}^{1/2}-\rho_{AB}^{1/2})}{2}^{1/2}
\\ \le & (d_C)^{1/4} \norm{(\sigma_{AB}-\rho_{AB}^{1/2})\sigma_{B}^{-\frac12}\sigma_{B}\sigma_{B}^{-\frac12}(\sigma_{AB}^{1/2}-\rho_{AB}^{1/2})}{2}^{1/2}
\\ \le & (d_C)^{1/4} \norm{(\sigma_{AB}^{1/2}-\rho_{AB}^{1/2})^2}{2}^{1/2}
\\ = & (d_C)^{1/4} \norm{\sigma_{AB}^{1/2}-\rho_{AB}^{1/2}}{4}^{1/2}
\end{align*}
For the second term,
\begin{align*}
&\norm{\sigma_{BC}^{\frac12}(\sigma_{B}^{-\frac12}-\rho_{B}^{-\frac12})\rho_{AB}^{1/2}}{4}
\\ \le & \norm{\sigma_{BC}^{\frac12}(\sigma_{B}^{-\frac12}-\rho_{B}^{-\frac12})\rho_{AB}(\sigma_{B}^{-\frac12}-\rho_{B}^{-\frac12})\sigma_{BC}^{\frac12}}{2}^{1/2}
\\ \le & d_A^{1/4} \norm{\sigma_{BC}^{\frac12}(\sigma_{B}^{-\frac12}-\rho_{B}^{-\frac12})\rho_{B}(\sigma_{B}^{-\frac12}-\rho_{B}^{-\frac12})\sigma_{BC}^{\frac12}}{2}^{1/2}\\
= & d_A^{1/4} \norm{\rho_{B}^{1/2}(\sigma_{B}^{-\frac12}-\rho_{B}^{-\frac12})\sigma_{BC}(\sigma_{B}^{-\frac12}-\rho_{B}^{-\frac12})\rho_{B}^{1/2}}{2}^{1/2}
\\ \le & d_A^{1/4}d_C^{1/4} \norm{\rho_{B}^{1/2}(\sigma_{B}^{-\frac12}-\rho_{B}^{-\frac12})\sigma_{B}(\sigma_{B}^{-\frac12}-\rho_{B}^{-\frac12})\rho_{B}^{1/2}}{2}^{1/2}
\\ = & d_A^{1/4}d_C^{1/4} \norm{(\sigma_{B}^{\frac12}-\rho_{B}^{\frac12})^2}{2}^{1/2}\\
=& d_A^{1/4}d_C^{1/4} \norm{\sigma_{B}^{\frac12}-\rho_{B}^{\frac12}}{4}
\end{align*}
\end{proof}

\begin{rem}{\rm The above lemma also applies to Schatten $p$-norm using a similar argument.
For $1\le p\le \infty$ and $\frac{1}{p}+\frac{1}{p'}$,
\begin{align*}
&\norm{\rho_{BC}^{\frac12}\rho_{B}^{-\frac12}\rho_{AB}\rho_{B}^{-\frac12}\rho_{BC}^{\frac12}
-\sigma_{BC}^{\frac12}\sigma_{B}^{-\frac12}\omega_{AB}\sigma_{B}^{-\frac12}\sigma_{BC}^{\frac12}}{p}\\
\le & \norm{\rho_{BC}^{\frac12}\rho_{B}^{-\frac12}\rho_{AB}^{\frac12}
-\sigma_{BC}^{\frac12}\sigma_{B}^{-\frac12}\sigma_{AB}^{\frac12}}{p}\\ \le &\ 2d_A^{1/2p'}\norm{\sigma_{BC}^{\frac12}-\rho_{BC}^{\frac12}}{2p}+2d_A^{1/2p'}d_C^{1/2p'}\norm{\sigma_{B}^{\frac12}-\rho_{B}^{\frac12}}{2p} +2d_C^{1/2p'}\norm{\omega_{AB}^{\frac12}-\rho_{AB}^{\frac12}}{2p}.
\end{align*}}
\end{rem}

\section{Quantum Markov Chain tomography}
In this section, we study the quantum Markov Chain tomography problem:

How many copies are needed to output $\hat \rho$ with expected trace distance $\epsilon$ (or infidelity $\delta$) to quantum Markov Chain $\rho_{ABC}$?

\subsection{Tripartite quantum Markov Chain}
We settle the complexity of this problem for a tripartite quantum system.

\vspace{2mm}
{\bf{Theorem}} \ref{main}.\
To solve the quantum Markov Chain tomography in system $\cH_{A}\ten \cH_{B}\ten\cH_{C} $ with $n$ samples, there are constant $c,C>0$ such that
\[
c\frac{d_B^2 (d_A^2+d_C^2)}{\delta} \leq n\leq C\min\{\frac {(d_A^2+d_C^2)d_B^2 \ln (d_Ad_Bd_C/\delta)}{\delta}, \frac{(d_A^2+d_C^2)d_B^2}{\delta^4}\}
\]
for infidelity $\delta$.
\[
c\frac{d_B^2 (d_A^2+d_C^2)}{\epsilon^2} \leq n\leq C\min\{\frac {(d_A^2+d_C^2)d_B^2 \ln (d_Ad_Bd_C/\epsilon)}{\epsilon^2}, \frac{(d_A^2+d_C^2)d_B^2}{\epsilon^4}\}
\]
for trace distance $\eps$.
\vspace{2mm}

\begin{proof}
Suppose the given state is $\rho_{ABC}$. The upper bound for infidelity $\delta$ is proved as follows:
\begin{itemize}
\item Use Lemma \ref{HHJ} to tomography $\rho_{AB}=\tr_C\rho_{ABC}$ and $\rho_{BC}=\tr_A\rho_{ABC}$ with infidelity $0.01\delta$, and the estimations are $\hat \rho_{AB}$ and $\hat\rho_{BC}$, respectively.
\item Output \[
\hat\rho_{BC}^{\frac12}\hat\rho_{B}^{-\frac12}\hat\rho_{AB}\hat\rho_{B}^{-\frac12}\hat\rho_{BC}^{\frac12}
\]
where $\hat\rho_B=\tr_C\rho_{BC}$.
\end{itemize}
By the monotonicity of fidelity under partial trace, we have
\[
F(\rho_B,\hat\rho_B)\geq F(\rho_{BC},\hat\rho_{BC})\geq 1-0.01\delta,
\]
According to Theorem \ref{petz:fidelity}, we have
\begin{align*}
F(\rho_{ABC},\hat\rho_{BC}^{\frac12}\hat\rho_{B}^{-\frac12}\hat\rho_{AB}\hat\rho_{B}^{-\frac12}\hat\rho_{BC}^{\frac12})\geq 1-0.18\delta.
\end{align*}
where we use the fact that for quantum Markov Chain $\rho_{ABC}=\rho_{BC}^{\frac12}\rho_{B}^{-\frac12}\rho_{AB}\rho_{B}^{-\frac12}\rho_{BC}^{\frac12}$.

In the first step, we use
\[
\frac{100 d_A^2d_B^2 \ln (d_Ad_B/\delta)}{\delta}+\frac{100 d_B^2d_C^2 \ln (d_Bd_C/\delta)}{\delta}=O(\frac {(d_A^2+d_C^2)d_B^2 \ln (d_Ad_Bd_C/\delta)}{\delta}).
\]
samples of $\rho_{ABC}$.

To prove the upper for distance $\epsilon$, we only need to choose $\delta=\epsilon^2/2$. The relation between $\epsilon$ and $\delta$ leads us to
\[
\frac{1}{2}||\rho_{ABC}-\hat\rho_{BC}^{\frac12}\hat\rho_{B}^{-\frac12}\hat\rho_{AB}\hat\rho_{B}^{-\frac12}\hat\rho_{BC}^{\frac12}||_1\leq \epsilon.
\]
The used number of samples is
\[
O(\frac {(d_A^2+d_C^2)d_B^2 \ln (d_Ad_Bd_C/\epsilon)}{\epsilon^2}).
\]
If we use Lemma \ref{OW16} with parameter $0.01\epsilon^2$ instead of Lemma \ref{HHJ} with parameter $0.01\delta$.
Theorem \ref{petz:distance} implies
\[
\frac{1}{2}||\rho_{ABC}-\hat\rho_{BC}^{\frac12}\hat\rho_{B}^{-\frac12}\hat\rho_{AB}\hat\rho_{B}^{-\frac12}\hat\rho_{BC}^{\frac12}||_1\leq 0.6\epsilon.
\]
We use
\[
O(\frac{(d_A^2+d_C^2)d_B^2}{\epsilon^4})
\]
samples.
The bound for $\delta$ leads to the following upper bound for $\delta$,
\[
O(\frac{(d_A^2+d_C^2)d_B^2}{\delta^4})
\]
The lower bound follows from the following fact: $\rho_{AB}\otimes \frac{id_C}{d_C}$ and $\frac{id_A}{d_A} \otimes \rho_{BC}$ and are both quantum Markov Chains for any $\rho_{AB}$ and $\rho_{BC}$. To tomography these quantum Markov Chains, we need to accomplish the tomography of $\rho_{AB}$ or $\rho_{BC}$.

According to the lower bound part of Lemma \ref{HHJ}, one needs at least
\[
\Omega(\frac{d_B^2\max\{d_A^2,d_C^2\}}{\epsilon^2})=\Omega(\frac{d_B^2 (d_A^2+d_C^2)}{\epsilon^2})
\]
samples of $\rho_{ABC}$.
The bound also leads to the following lower bound for $\delta$,
\[
\Omega(\frac{d_B^2 (d_A^2+d_C^2)}{\delta}).
\]
\end{proof}
\begin{rem}{\rm The above theorem is robust. We only need $\rho_{ABC}$ is $0.01\epsilon$ (or $0.01\delta$) close to a quantum Markov Chain.}
\end{rem}

\subsection{Multipartite quantum Markov Chain}

In this subsection, we generalized the tomography results into multipartite quantum Markov Chain of the system $\otimes_{i=1}^m \cH_i$. Here $\rho$ is called a quantum Markov Chain in $\otimes_{i=1}^m \cH_i$ if
\begin{align}
\rho=\mathcal{N}_{m-1}\circ\cdots\circ\mathcal{N}_2 (\rho_{1,2})
\end{align}
where $\mathcal{N}_{i}: \cB(\cH_{i})\mapsto \cB(\cH_{i}\otimes\cH_{i+1})$ with form $\mathcal{N}_{i}(X)=\rho_{i,i+1}^{1/2}(\rho_i^{-1/2}X\rho_i^{-1/2}\otimes I_{i+1})\rho_{i,i+1}^{1/2}$, the Petz map defined of system $\cH_{i}\otimes\cH_{i+1}$ with $\rho_{s}$ denotes the reduced density matrix of system $s\subseteq\{1,2,\cdots,m\}$.

\vspace{2mm}
{\bf{Theorem}} \ref{main-4}.\
For quantum system $\otimes_{i=1}^m \cH_i$ with $d_i$ be the dimension of the $i$-th subsystem,
$\tilde{O}(\frac{m^2\max_{i} \{d_i^2d_{i+1}^2\}}{\delta})$ copies are sufficient to tomography a quantum Markov Chain with infidelity $\delta$.
\vspace{2mm}

\begin{proof}
We use the following algorithm.
\begin{itemize}
\item For each even number $2\leq i\leq m-1$, we use Lemma \ref{HHJ} to tomography $\rho_{i,i+1}$ in parallel, with infidelity $\delta_i$ and a successful probability $1-\frac{1}{100m}$.

\item For each odd number $2\leq i\leq m-1$, we use Lemma \ref{HHJ} to tomography $\rho_{i,i+1}$ in parallel, with infidelity $\delta_i$ and a succesful probability $1-\frac{1}{100m}$.

\item Let the estimations be $\hat \rho_{i,i+1}$ and $\hat{\mathcal{N}}_{i}: \cB(\cH_{i})\mapsto \cB(\cH_{i}\otimes\cH_{i+1})$ with form
\[\hat{\mathcal{N}}_{i}(X)=\hat\rho_{i,i+1}^{1/2}(\hat\rho_i^{-1/2}X\hat\rho_i^{-1/2}\otimes id_{i+1})\hat\rho_{i,i+1}^{1/2}.\]

\item Output \[
\hat\rho=\hat{\mathcal{N}}_{m-1}\circ\cdots\circ\hat{\mathcal{N}}_2 (\hat\rho_{1,2}).
\]
\end{itemize}
The parameters $\delta_i$ will be chosen later.

By union bound, we know that with probability at least $1-\frac{m}{100m}=0.99$,
\[
1-F(\rho_{i,i+1},\hat\rho_{i,i+1})\leq \delta_i\ \ \ \ \forall \ \ \ 1\leq i\leq m-1.
\]
We first show the following,
\begin{align*}
&||\rho_{m-1,m}^{\frac12}\rho_{m-1}^{-\frac12}\cdots \rho_{1,2}^{\frac12}-\hat\rho_{m-1,m}^{\frac12}\hat\rho_{m-1}^{-\frac12}\cdots \hat\rho_{1,2}^{\frac12}||_2\\
\leq & ||\rho_{m-1,m}^{\frac12}\rho_{m-1}^{-\frac12}\cdots \rho_{1,2}^{\frac12}-\hat\rho_{m-1,m}^{\frac12}\rho_{m-1}^{-\frac12}\cdots \rho_{1,2}^{\frac12}||_2+||\hat\rho_{m-1,m}^{\frac12}\hat\rho_{m-1}^{-\frac12}\cdots \hat\rho_{1,2}^{\frac12}-\hat\rho_{m-1,m}^{\frac12}\rho_{m-1}^{-\frac12}\cdots \rho_{1,2}^{\frac12}||_2\\
=& \sqrt{\tr\mathcal{N}_1\circ\cdots \circ\mathcal{N}_{m-2}((\rho_{m-1,m}^{\frac12}-\hat\rho_{m-1,m}^{\frac12})^2)}+||\hat\rho_{m-1,m}^{\frac12}\hat\rho_{m-1}^{-\frac12}\cdots \hat\rho_{1,2}^{\frac12}-\hat\rho_{m-1,m}^{\frac12}\rho_{m-1}^{-\frac12}\cdots \rho_{1,2}^{\frac12}||_2\\
=& ||\rho_{m-1,m}^{\frac12}-\hat\rho_{m-1,m}^{\frac12}||_2+||\hat\rho_{m-1,m}^{\frac12}\hat\rho_{m-1}^{-\frac12}\cdots \hat\rho_{1,2}^{\frac12}-\hat\rho_{m-1,m}^{\frac12}\rho_{m-1}^{-\frac12}\cdots \rho_{1,2}^{\frac12}||_2\\
\leq& ||\rho_{m-1,m}^{\frac12}-\hat\rho_{m-1,m}^{\frac12}||_2+||\hat\rho_{m-1,m}^{\frac12}\rho_{m-1}^{-\frac12}\rho_{m-2,m-1}^{\frac12}\cdots \rho_{1,2}^{\frac12}-\hat\rho_{m-1,m}^{\frac12}\hat\rho_{m-1}^{-\frac12}\rho_{m-2,m-1}^{\frac12}\cdots \rho_{1,2}^{\frac12}||_2\\
&\ \ \ \ \ \ \ \ \ \ \ \ \ \ \ \ \ \ \ \ \ \ \ \ \ \ \ \ \ \ \ \ \ \ \ +||\hat\rho_{m-1,m}^{\frac12}\hat\rho_{m-1}^{-\frac12}\hat\rho_{m-2,m-1}^{\frac12}\cdots \hat\rho_{1,2}^{\frac12}-\hat\rho_{m-1,m}^{\frac12}\hat\rho_{m-1}^{-\frac12}\rho_{m-2,m-1}^{\frac12}\cdots \rho_{1,2}^{\frac12}||_2\\
=& ||\rho_{m-1,m}^{\frac12}-\hat\rho_{m-1,m}^{\frac12}||_2+\sqrt{\tr \mathcal{N}_1\circ\cdots \circ\mathcal{N}_{m-3}((\hat\rho_{m-1,m}^{\frac12}\rho_{m-1}^{-\frac12}\rho_{m-2,m-1}^{\frac12}-\hat\rho_{m-1,m}^{\frac12}\hat\rho_{m-1}^{-\frac12}\rho_{m-2,m-1}^{\frac12})^2)}\\
&\ \ \ \ \ \ \ \ \ \ \ \ \ \ \ \ \ \ \ \ \ \ \ \ \ \ \ \ \ \ \ \ \ \ \ +||\hat\rho_{m-1,m}^{\frac12}\hat\rho_{m-1}^{-\frac12}\hat\rho_{m-2,m-1}^{\frac12}\cdots \hat\rho_{1,2}^{\frac12}-\hat\rho_{m-1,m}^{\frac12}\hat\rho_{m-1}^{-\frac12}\rho_{m-2,m-1}^{\frac12}\cdots \rho_{1,2}^{\frac12}||_2\\
=& ||\rho_{m-1,m}^{\frac12}-\hat\rho_{m-1,m}^{\frac12}||_2+||\hat\rho_{m-1,m}^{\frac12}\rho_{m-1}^{-\frac12}\rho_{m-2,m-1}^{\frac12}-\hat\rho_{m-1,m}^{\frac12}\hat\rho_{m-1}^{-\frac12}\rho_{m-2,m-1}^{\frac12}||_2\\
&\ \ \ \ \ \ \ \ \ \ \ \ \ \ \ \ \ \ \ \ \ \ \ \ \ \ \ \ \ \ \ \ \ \ \ +||\hat\rho_{m-1,m}^{\frac12}\hat\rho_{m-1}^{-\frac12}\hat\rho_{m-2,m-1}^{\frac12}\cdots \hat\rho_{1,2}^{\frac12}-\hat\rho_{m-1,m}^{\frac12}\hat\rho_{m-1}^{-\frac12}\rho_{m-2,m-1}^{\frac12}\cdots \rho_{1,2}^{\frac12}||_2\\
=& ||\rho_{m-1,m}^{\frac12}-\hat\rho_{m-1,m}^{\frac12}||_2+||\rho_{m-1}^{\frac12}-\hat\rho_{m-1}^{\frac12}||_2\\
&\ \ \ \ \ \ \ \ \ \ \ \ \ \ \ \ \ \ \ \ \ \ \ \ \ \ \ \ \ \ \ \ \ \ \ +\sqrt{\tr \mathcal{N}_{m-1}((\hat\rho_{m-2,m-1}^{\frac12}\cdots \hat\rho_{1,2}^{\frac12}-\rho_{m-2,m-1}^{\frac12}\cdots \rho_{1,2}^{\frac12})^2)}\\
=&||\rho_{m-1,m}^{\frac12}-\hat\rho_{m-1,m}^{\frac12}||_2+||\rho_{m-1}^{\frac12}-\hat\rho_{m-1}^{\frac12}||_2+||\hat\rho_{m-2,m-1}^{\frac12}\cdots \hat\rho_{1,2}^{\frac12}-\rho_{m-2,m-1}^{\frac12}\cdots \rho_{1,2}^{\frac12}||_2\\
&\cdots\\
\leq& \sum_{i=1}^{m-1}||\rho_{i,i+1}^{\frac12}-\hat\rho_{i,i+1}^{\frac12}||_2+\sum_{i=2}^{m-1}||\rho_{i}^{\frac12}-\hat\rho_{i}^{\frac12}||_2\\
\leq & 2\sum_{i=1}^{m-1}\sqrt{1-F(\rho_{i,i+1},\hat\rho_{i,i+1})}+2\sum_{i=2}^{m-1}\sqrt{1-F(\rho_{i},\hat\rho_i)}\\
\leq & 2\sum_{i=1}^{m-1}\sqrt{1-F(\rho_{i,i+1},\hat\rho_{i,i+1})}+2\sum_{i=2}^{m-1}\sqrt{1-F(\rho_{i,i+1},\hat\rho_{i,i+1})}\\
\leq &4\sum_{i=1}^{m-1}\sqrt{1-F(\rho_{i,i+1},\hat\rho_{i,i+1})}\\
\leq &4 \sum_{i=1}^{m-1} \delta_i^{\frac12}.
\end{align*}
According to Lemma \ref{f-1}, we have
\[
F(\rho,\hat\rho)\geq 1-8(\sum_{i=1}^{m-1} \delta_i^{\frac12})^2.
\]
To ensure
\[
F(\rho,\hat\rho)\geq 1-\delta,
\]
one only needs to choose $\delta_i$ such that
\[
\sum_{i=1}^{m-1} \delta_i^{\frac12}\leq \frac{\delta^{\frac12}}{2\sqrt{2}}.
\]
For instance, we can choose $\delta_i=\frac{\delta}{8m^2}$.

The parallel processing in the first two steps allows us to use the same copies to accomplish the tomography of $\rho_{1,2}$, $\rho_{3,4}$, $\cdots$, as well as $\rho_{2,3}$, $\rho_{4,5}$, $\cdots$.
This follows from the observation that: For each sample $i_1i_2\cdots i_m$ of an $m$-partite distribution $p_{1,2,\cdots,m}$, $i_k$ is a sample of $p_k$ for each $k$.

According to Lemma \ref{HHJ}, the total number of samples is
\[
O(\frac{\max_i\{d_i^2d_{i+1}^2\ln m \ln(d_id_{i+1}/\delta_i)\}}{\delta_i})=\tilde O(\frac{m^2\max_i\{d_i^2d_{i+1}^2\}}{\delta}).
\]
\end{proof}

\begin{rem}
For some region of parameters, $d_i$ and $m$, one can do better than $\tilde O(\frac{m^2\max_i\{d_i^2d_{i+1}^2\}}{\delta})$. For instance, $\max_i\{d_i^2d_{i+1}^2\}$ is siginificantly larger than other $d_jd_{j+1}$, then we can choose not uniform $\delta_i$ to decrease the number of samples.
\end{rem}

\section{Certification of quantum Markov Chain}

In this section, we study the quantum Markov Chain certification problem:

For a quantum Markov Chain $\rho_{ABC}$ and a known quantum Markov Chain $\sigma_{ABC}$, how many copies of $\rho_{ABC}$ are needed to distinguish between two cases:
\begin{itemize}
\item $\rho_{ABC}=\sigma_{ABC}$;
\item $ F(\rho_{ABC},\sigma_{ABC})<1- \delta$ (or, $||\rho_{ABC}-\sigma_{ABC}||_1>\epsilon$)?
\end{itemize}

\vspace{2mm}

{\bf{Theorem}} \ref{main-2}. \
The sample complexity of quantum Markov Chain certification is ${\Theta}(\frac{(d_A+d_C)d_B}{\delta})$ for indfidelity $\delta$, and ${\Theta}(\frac{(d_A+d_C)d_B}{\epsilon^2})$ for trace distance $\epsilon$.
\vspace{2mm}

\begin{proof}
The upper bound for infidelity $\delta$ is proved as follows:
\begin{itemize}
\item Use Lemma \ref{BOW17} certify whether $\rho_{AB}=\sigma_{AB}$ and $\rho_{BC}=\sigma_{BC}$ with infidelity $0.01\delta$ and success probability $0.99$. and the estimations are $\hat \rho_{AB}$ and $\hat\rho_{BC}$, respectively.
\item If any of the certifications returns that the infidelity is at least $0.01\delta$, then we output ``They are at least $\delta$-far". Otherwise, return $\rho_{ABC}=\sigma_{ABC}$.
\end{itemize}
If one of the certification returns the infidelity at least $0.01\delta$, says $F(\rho_{AB},\sigma_{AB})<1-0.01\delta$, then $\rho_{ABC}\neq\sigma_{ABC}$, which implies $ F(\rho_{ABC},\sigma_{ABC})<1- \delta$.

If both $F(\rho_{AB},\sigma_{AB})>1-0.01\delta$ and $F(\rho_{BC},\sigma_{BC})>1-0.01\delta$.
By the monotonicity of fidelity under partial trace, we have
\[
F(\rho_B,\hat\rho_B)\geq F(\rho_{BC},\hat\rho_{BC})\geq 1-0.01\delta,
\]
Corrolary \ref{qmc:fidelity} leads to
\[
F(\rho_{ABC},\sigma_{ABC})>1-0.18\delta.
\]
This implies $\rho_{ABC}=\sigma_{ABC}$.

In this process, we use
\[
{\Theta}(\frac{(d_A+d_C)d_B}{\delta})
\]
samples of $\rho_{ABC}$.

To prove the upper for distance $\epsilon$, we only need to choose $\delta=\epsilon^2/2$ and use
\[
{\Theta}(\frac{(d_A+d_C)d_B}{\epsilon^2})
\]
samples of $\rho_{ABC}$.

The lower bound follows from the following fact: $\rho_{AB}\otimes \frac{I_C}{d_C}$ and $\frac{I_A}{d_A} \otimes \rho_{BC}$ and are both quantum Markov Chains for any $\rho_{AB}$ and $\rho_{BC}$. We let $\sigma=\frac{I_A}{d_A}\ten\frac{I_B}{d_B}\otimes \frac{I_C}{d_C}$.
If we know $\rho_{ABC}=\rho_{AB}\otimes \frac{I_C}{d_C}$, then
\[
||\rho_{ABC}-\sigma_{ABC}||_1=||\rho_{AB}- \frac{I_A}{d_A}\ten\frac{I_B}{d_B}||_1.
\]
According to Lemma \ref{BOW17}, $\Omega(\frac{d_Ad_B}{\epsilon^2})$ copies of $\rho_{AB}$ is needed to distinguish
\begin{itemize}
\item $\rho_{AB}=\sigma_{AB}$;
\item $||\rho_{AB}-\sigma_{AB}||_1>\epsilon$.
\end{itemize}
By the relation between $\epsilon$ and $\delta$, we know that
$\Omega(\frac{d_Ad_B}{\delta})$ copies of $\rho_{AB}$ is needed to distinguish
\begin{itemize}
\item $\rho_{AB}=\sigma_{AB}$;
\item $F(\rho_{AB},\sigma_{AB})<1- \delta$.
\end{itemize}
By choosing $\rho_{ABC}=\frac{id_A}{d_A} \otimes \rho_{BC}$, similar arguments lead to lower bound
$\Omega(\frac{d_Bd_C}{\epsilon^2})$ and $\Omega(\frac{d_Bd_C}{\delta})$. This completes the proof.
\end{proof}
\section{Testing quantum Markov Chain}
In this section, we study the problem of testing quantum Markov Chain, i.e., distinguishing between two cases for $\rho_{ABC}$:
\begin{itemize}
\item $\rho_{ABC}$ is a quantum Markov Chain;
\item $ \norm{\rho_{ABC}-\rho_{BC}^{\frac12}\rho_{B}^{-\frac12}\rho_{AB}\rho_{B}^{-\frac12}\rho_{BC}^{\frac12}}{1}\geq \epsilon$.
\end{itemize}
We use the following observation to study the quantum Markov Chain testing problem.
\begin{cor}\label{cor:L2-new}
Let $\rho_{B}=\tr_A\rho_{AB}$, and $\sigma_B=\tr_C\sigma_{BC}$ with
\[F(\rho_{BC},\sigma_{BC})\ge 1-\delta,\]
then
\[ \norm{\rho_{BC}^{\frac12}\rho_{B}^{-\frac12}\rho_{AB}\rho_{B}^{-\frac12}\rho_{BC}^{\frac12}-\sigma_{BC}^{\frac12}\sigma_{B}^{-\frac12}\rho_{AB}\sigma_{B}^{-\frac12}\sigma_{BC}^{\frac12}}{2}\le 8\delta^{\frac12}.\]
\end{cor}
\begin{proof}
We first observe
\[
F(\rho_B,\sigma_B)\geq F(\rho_{BC},\sigma_{BC})\geq 1-\delta.
\]
Theorem \ref{petz:fidelity} implies
\[
F(\rho_{BC}^{\frac12}\rho_{B}^{-\frac12}\rho_{AB}\rho_{B}^{-\frac12}\rho_{BC}^{\frac12},\sigma_{BC}^{\frac12}\sigma_{B}^{-\frac12}\rho_{AB}\sigma_{B}^{-\frac12}\sigma_{BC}^{\frac12})\ge 1-8\delta.
\]
According to the relation between $\ell_2$ norm, $\ell_1$-norm and infidelity, we have
\[ \norm{\rho_{BC}^{\frac12}\rho_{B}^{-\frac12}\rho_{AB}\rho_{B}^{-\frac12}\rho_{BC}^{\frac12}-\sigma_{BC}^{\frac12}\sigma_{B}^{-\frac12}\rho_{AB}\sigma_{B}^{-\frac12}\sigma_{BC}^{\frac12}}{2}\leq \norm{\rho_{BC}^{\frac12}\rho_{B}^{-\frac12}\rho_{AB}\rho_{B}^{-\frac12}\rho_{BC}^{\frac12}-\sigma_{BC}^{\frac12}\sigma_{B}^{-\frac12}\rho_{AB}\sigma_{B}^{-\frac12}\sigma_{BC}^{\frac12}}{1} \le 8\delta^{\frac12}.\]
\end{proof}
Full-state tomography is not always necessary for the quantum Markov Chain testing problem.

\vspace{2mm}
{\bf{Theorem}} \ref{main-3}. \
$\tilde{O}(\frac{\min\{d_Ad_B^3d_C^3,d_A^3d_B^3d_C\}}{\epsilon^2})$ copies are sufficient to certify whether
$\rho_{ABC}$ is a quantum Markov Chain or $\epsilon$-far from its Petz recovery state for trace distance $\epsilon$.
\vspace{2mm}
\begin{proof}
We use the following tester to show the upper bound $\tilde{O}(\frac{d_Ad_B^3d_C^3}{\epsilon^2})$.
\begin{itemize}
\item Use Lemma \ref{HHJ} to tomography $\rho_{BC}=\tr_A\rho_{ABC}$ with infidelity $\delta:=\frac{\epsilon^2}{400d_Ad_Bd_C}$, and the estimation is $\sigma_{BC}$.
\item Apply Petz map $\sigma_{BC}^{\frac12}\sigma_{B}^{-\frac12}\cdot\sigma_{B}^{-\frac12}\sigma_{BC}^{\frac12}$ on $\rho_{AB}$ to obtain $\sigma_{BC}^{\frac12}\sigma_{B}^{-\frac12}\rho_{AB}\sigma_{B}^{-\frac12}\sigma_{BC}^{\frac12}$.
\item Use Lemma \ref{BOW17} to distinguish $\rho_{ABC}$ and
$\sigma_{BC}^{\frac12}\sigma_{B}^{-\frac12}\rho_{AB}\sigma_{B}^{-\frac12}\sigma_{BC}^{\frac12}$ in $\ell_2$ norm.
Return ``$\rho_{ABC}$ is a quantum Markov China'' if $||\rho_{ABC}-\sigma_{BC}^{\frac12}\sigma_{B}^{-\frac12}\rho_{AB}\sigma_{B}^{-\frac12}\sigma_{BC}^{\frac12}||_2<0.4 \frac{\epsilon}{\sqrt{d_Ad_Bd_C}}$. Otherwise, return ``$ \norm{\rho_{ABC}-\rho_{BC}^{\frac12}\rho_{B}^{-\frac12}\rho_{AB}\rho_{B}^{-\frac12}\rho_{BC}^{\frac12}}{1}\geq \eps$''.
\end{itemize}
Corollary \ref{cor:L2-new} implies
\[ \norm{\rho_{BC}^{\frac12}\rho_{B}^{-\frac12}\rho_{AB}\rho_{B}^{-\frac12}\rho_{BC}^{\frac12}-\sigma_{BC}^{\frac12}\sigma_{B}^{-\frac12}\rho_{AB}\sigma_{B}^{-\frac12}\sigma_{BC}^{\frac12})}{2}\le 8\delta^{\frac12}=\frac{2\epsilon}{5\sqrt{d_Ad_Bd_C}}.\]
If $\rho_{ABC}$ is a quantum Markov Chain, we have
\[ \norm{\rho_{ABC}-\sigma_{BC}^{\frac12}\sigma_{B}^{-\frac12}\rho_{AB}\sigma_{B}^{-\frac12}\sigma_{BC}^{\frac12}}{2}\le \frac{2\epsilon}{5\sqrt{d_Ad_Bd_C}}.\]
If
\[\norm{\rho_{ABC}-\rho_{BC}^{\frac12}\rho_{B}^{-\frac12}\rho_{AB}\rho_{B}^{-\frac12}\rho_{BC}^{\frac12}}{1}\geq \epsilon,\]
then
\[\norm{\rho_{ABC}-\rho_{BC}^{\frac12}\rho_{B}^{-\frac12}\rho_{AB}\rho_{B}^{-\frac12}\rho_{BC}^{\frac12}}{2}\geq \frac{\epsilon}{\sqrt{d_Ad_Bd_C}}.\]
Then
\[ \norm{\rho_{ABC}-\sigma_{BC}^{\frac12}\sigma_{B}^{-\frac12}\rho_{AB}\sigma_{B}^{-\frac12}\sigma_{BC}^{\frac12}}{2}\ge \frac{3\epsilon}{5\sqrt{d_Ad_Bd_C}}.\]
In this tester, we use
\[
O(\frac{d_B^2d_C^2 \ln (d_Bd_C/\delta)}{\delta}+\frac{d_Ad_Bd_C}{\epsilon^2})=\tilde O(\frac{d_Ad_B^3d_C^3}{\epsilon^2})
\]
samples of $\rho_{ABC}$.

The upper bound $\tilde{O}(\frac{d_Cd_A^3d_B^3}{\epsilon^2})$ follows from the following fact of quantum Markov Chain: For quantum Markov Chain $\rho_{ABC}$,
\[
\rho_{ABC}=\rho_{BC}^{\frac12}\rho_{B}^{-\frac12}\rho_{AB}\rho_{B}^{-\frac12}\rho_{BC}^{\frac12}=\rho_{AB}^{\frac12}\rho_{B}^{-\frac12}\rho_{BC}\rho_{B}^{-\frac12}\rho_{AB}^{\frac12}.
\]
\end{proof}

\section{Acknowledgement}
We want to thank Andreas Winter and Marco Tomamichel for their helpful discussions.
\bibliographystyle{alpha}
\bibliography{opt-tomo}

\end{document}